\documentclass[12pt,reqno]{amsart}
\usepackage[top=1in, bottom=1in, left=1in, right=1in]{geometry}
\usepackage{times}
\usepackage{listings}
\usepackage{graphicx}
\usepackage{epstopdf}

\usepackage[usenames,dvipsnames]{color}
\usepackage{tikz}
\usepackage{ifthen}
\usepackage{tikz-3dplot}
\usepackage{pgfplots}
\usepackage{amssymb}
\usepackage{bbm}
\usepackage{wasysym}
\usepackage{enumitem}
\usepackage{amsthm}
\usepackage{amsmath,amssymb,mathrsfs}
\usepackage{hyperref}
\usepackage{graphicx}
\usepackage[arrow,matrix,curve]{xy}
\usepackage{lineno}
\usepackage{algorithmicx}
\usepackage{algpseudocode}

\usepackage{dcolumn} 

\usepackage{color}

\definecolor{codegreen}{rgb}{0,0.6,0}
\definecolor{codegray}{rgb}{0.5,0.5,0.5}
\definecolor{codepurple}{rgb}{0.58,0,0.82}
\definecolor{codeblue}{rgb}{0,0,0.7}
\definecolor{backcolour}{rgb}{0.96,0.96,0.93}

\lstdefinestyle{mystyle}{
    backgroundcolor=\color{backcolour},
    commentstyle=\color{codegreen},
    keywordstyle=\color{magenta},
    numberstyle=\tiny\color{codegray},
    stringstyle=\color{codepurple},
    basicstyle=\ttfamily\footnotesize,
    breakatwhitespace=false,
    breaklines=true,
    captionpos=b,
    keepspaces=true,
    numbers=left,
    numbersep=5pt,
    showspaces=false,
    showstringspaces=false,
    showtabs=false,
    tabsize=2
}

\newcommand{\COMMENT}[1]{{\color{blue}\texttt{ /* #1 */}}}
\newcommand{\INPUT}[1]{\textbf{Input:} #1}
\newcommand{\OUTPUT}[1]{\textbf{Output:} #1}

\numberwithin{equation}{section}

\newtheorem{theorem}{Theorem}[section]

\newtheorem{lemma}[theorem]{Lemma}

\theoremstyle{definition}
\newtheorem{algorithm}[theorem]{Algorithm}
\newtheorem{definition}[theorem]{Definition}

\newtheorem{remark}[theorem]{Remark}

\newtheorem{problem}[theorem]{Problem}

\newcommand\CC{{\mathbb C}}
\newcommand\NN{{\mathbb N}}
\newcommand\RR{{\mathbb R}}
\newcommand\ZZ{{\mathbb Z}}
\newcommand\cA{{\mathcal A}}

\newcommand\cG{{\mathcal G}}

\newcommand\sA{{\mathscr A}}
\newcommand\sL{{\mathscr L}}
\newcommand\sU{{\mathscr U}}

\newcommand{\C}{\mathbb{C}}
\newcommand{\N}{\mathbb{N}}
\newcommand{\Q}{\mathbb{Q}}
\newcommand{\R}{\mathbb{R}}

\newcommand{\Z}{\mathbb{Z}}

\newcommand{\zz}{{\mathbf{z}}}
\newcommand{\ww}{{\mathbf{w}}}
\newcommand{\bb}{{\mathbf{b}}}
\newcommand{\xx}{{\mathbf{x}}}
\DeclareMathOperator{\pr}{pr}

\newcommand{\eps}{\varepsilon}

\newcommand{\alp}{\alpha}

\newcommand{\Log}{{\operatorname{Log}}}

\newcommand{\Res}{{\operatorname{Res}}}

\newcommand\New{{\rm New}}

\newcommand\Var{{\rm Var}}
\newcommand\cres{{\rm CycRes}}

\def\supp{{\rm supp}}

\DeclareMathOperator{\ord}{ord}

\DeclareMathOperator{\CycResult}{CycResult}
\DeclareMathOperator{\Multiplier}{Multiplier}

\newcommand{\bz}{\mathbf{z}}

\def\endrk{\hfill$\hexagon$}

\newcommand\minus{\smallsetminus}

\newcommand{\struc}[1]{{\color{blue}  #1}}

\begin{document}
\mbox{}
\vspace*{-12mm}
\title{Lopsided Approximation of Amoebas}

\author{Jens Forsg{\aa}rd}
\address{Department of Mathematics \\
Texas A\&M University \\ College Station, TX 77843.}
\email{jensf@math.tamu.edu}

\author{Laura Felicia Matusevich}
\email{laura@math.tamu.edu}

\author{Nathan Mehlhop}
\email{mehl144@tamu.edu}

\author{Timo de Wolff}
\email{dewolff@math.tamu.edu}

\thanks{
LFM was partially supported by NSF Grant DMS 1500832.
}


\begin{abstract}
The amoeba of a Laurent polynomial is the image of the
corresponding hypersurface under the coordinatewise
log absolute value map.
In this article, we demonstrate that a theoretical amoeba approximation 
method due to Purbhoo can be used efficiently in practice. To do
this, we resolve the main bottleneck in Purbhoo's method by exploiting
relations between cyclic 
resultants. We use the same approach to give an approximation of the Log 
preimage of the amoeba of a Laurent polynomial using semi-algebraic sets.
We also provide
a \textsc{SINGULAR}/\textsc{Sage} implementation of these algorithms,
which shows a significant speedup when our specialized cyclic
resultant computation is used, versus a general purpose resultant algorithm.

\end{abstract}
\maketitle

\lstset{style=mystyle}

\section{Introduction}
\label{sec:introduction}

Consider a \struc{\emph{Laurent polynomial}}
\begin{equation}
\label{eqn:f}
f(\mathbf{z}) = \sum_{j = 1}^d b_j \,\mathbf{z}^{\alp(j)},
\end{equation}
where $\struc{\mathbf{z}} := (z_1,\ldots,z_n)$. We denote by $\struc{\Var(f)}$
the hypersurface defined by $f$ in the maximal open torus
$\struc{(\C^*)^n} := (\C \setminus \{0\})^n$ of $\CC^n$.

\begin{definition}
\label{def:amoeba}
The \struc{\emph{log absolute value map}} is given by
\begin{eqnarray}
\label{Equ:LogMap}
	\struc{\Log \ |\cdot|}\colon \ \left(\C^*\right)^n \to \R^n, \quad (z_1,\ldots,z_n) \mapsto (\log|z_1|,
\ldots, \log|z_n|) \, .
\end{eqnarray}

The \struc{\textit{amoeba}} $\struc{\sA(f)}$ of $f$ is defined as $\Log \ |\Var(f)|$. The
\struc{\textit{unlog amoeba}} $\struc{\sU(f)}$ of $f$ is 
defined as $|\Var(f)|$, where
\begin{eqnarray*}
	\struc{|\cdot|}\colon \ \left(\C^*\right)^n \to \R^n, \quad (z_1,\ldots,z_n) \mapsto (|z_1|, \ldots, |z_n|) \, .
\end{eqnarray*}
\endrk
\end{definition}

Gelfand, Kapranov and Zelevinsky introduced amoebas in~\cite[Definition~6.1.4]{Gelfand:Kapranov:Zelevinsky}
in the context of toric geometry. Since then, amoebas have been used
in different areas such as
complex analysis \cite{Forsberg:Passare:Tsikh,Passare:Rullgard:Spine}, real algebraic
curves \cite{Mikhalkin:Annals}, statistical thermodynamics \cite{Passare:Pochekutov:Tsikh},
and nonnegativity of real polynomials
\cite{Iliman:deWolff:Circuits}.
Overviews on amoeba theory include \cite{deWolff:AmoebaTropicalizationSurvey,Mikhalkin:Survey,Passare:Tsikh:Survey}.\\

The usefulness of amoebas motivates the problem of finding an
efficient algorithm for computing them. 
Since $\Log|\cdot|$ is a non-algebraic map, $\sA(f)$ is not a semi-algebraic set in general. The amoeba $\sA(f)$ is, however, a
semi-analytic set.
The absolute value map, on the other hand, is a real algebraic map.
Therefore, the unlog amoeba $\sU(f)$ is a real semi-algebraic set. 
Hence, the ideal
solution to the problem of amoeba computation can be described as follows.

\begin{quote}
Given a Laurent polynomial $f \in \C[\mathbf{z}^{\pm 1}]$, efficiently
compute a real semi-algebraic description of the unlog amoeba
$\sU(f)$.
\end{quote}

Theobald was the first to tackle computational aspects of amoebas in~\cite{Theobald:ComputingAmoebas}. He described, in the case $n=2$,
methods to compute a superset of the boundary of $\sA(f)$ called the
\struc{\textit{contour}} of $\sA(f)$. Later, his method
was extended by Schroeter and the fourth author
in~\cite{Schroeter:deWolff:Boundary}, who provide a method to test
whether a contour point of $\sA(f)$ belongs to the boundary of
$\sA(f)$ in the case when the hypersurface defined by $f$ is smooth.

An approach, different than Theobald's, for computing amoebas arises from following problem.

\begin{problem}[The Membership Problem]
Let $f \in \C[\mathbf{z}^{\pm 1}]$ and $|\mathbf{v}| \in
\R^n$. Provide a certificate $C$ such that if $C(|\mathbf{v}|)$ is
true, then $\Log|\mathbf{v}| \notin \sA(f)$.
\label{Prob:MembershipProblem}
\end{problem}

The membership problem was addressed by Purbhoo in~\cite{Purbhoo},
using the notion of lopsidedness; see Definition~\ref{def:lopsided}.
A second approach for certifying non-membership of a point in $\sA(f)$
was provided by Theobald and the fourth author in
\cite{Theobald:deWolff:SOS} using semidefinite programming and sums of
squares. A rough but quick method to approximate amoebas was given by
Avenda\~{n}o, Kogan, Nisse and Rojas \cite{Nisse:Rojas:et:al} via
tropical geometry. In~\cite{BKS16}, Bogdanov, Kytmanov and Sadykov
computed amoebas in dimensions two and three, with focus on amoebas with
the most complicated topology possible. We remark that even for univariate polynomials deciding the membership problem is NP-hard \cite[Section 1.2]{Nisse:Rojas:et:al}.

The main aim of this article is to make the results in~\cite{Purbhoo} effective and efficient in practice.

\begin{definition}[{\cite[Page 24]{Rullgard:Diss}} and {\cite[Definition~1.2]{Purbhoo}}]
\label{def:lopsided}
Let $f$ be as in \eqref{eqn:f}.
We say that $f$ is \struc{\textit{lopsided}} at a point
$\Log|\mathbf{v}| \in \R^n$
if there exists an index $k \in \{1,\ldots,d\}$ such that
\begin{eqnarray}
  \left|b_k\, \mathbf{v}^{\alp(k)}\right| & > & \sum_{j \in \{1,\ldots,d\} \setminus \{k\}} \left|b_j\, \mathbf{z}^{\alp(j)}\right|. \label{Equ:LopsidedEquation}
\end{eqnarray}

The \struc{\textit{lopsided amoeba}} of $f$ is defined by
\begin{eqnarray*}
 \struc{\sL(f)} & := & \big\{\Log|\mathbf{v}| \in \R^n \ \big| \
                       f \text{ is not lopsided at } \Log|\mathbf{v}| \big\}.
\end{eqnarray*}
\endrk
\end{definition}

It is not hard to see that $\sA(f) \subseteq \sL(f)$; the special case
$n=1$ of this result was proven in the nineteenth century, in
an equivalent formulation, by Pellet~\cite{pellet}. 
We remark that, in general, $\sA(f) \neq \sL(f)$.

Following Purbhoo, we introduce
\begin{eqnarray}
& & \struc{\cres(f;r)} := \ \prod_{k_1 = 0}^{r-1} \cdots \prod_{k_n = 0}^{r-1} f\left(e^{2\pi i k_1 / r} z_1,\ldots,e^{2\pi i k_n / r} z_n\right) \label{EquIterPolynLops} \\
  & & \qquad  =   \Res_{u_n} \left(\Res_{u_{n-1}} \left(\ldots \Res_{u_1}(f(u_1 z_1,\ldots,u_n z_n),u_1^r - 1),\ldots, u_{n-1}^r - 1\right),u_n^r - 1\right),\nonumber
\end{eqnarray}
where ${\struc{\Res_x(g,h)}}$ denotes the resultant of the polynomials
$g$ and $h$ with respect to the variable $x$.
If $g=g(x)$ is a univariate polynomial, then
$\Res_x(g,x^r-1)$ is called a \struc{\textit{cyclic resultant}} of $g$.

Note that $\cres(f;r)$ is a Laurent
polynomial in $z_1,\dots,z_n$. The informal version of~\cite[Theorem
1]{Purbhoo} states that
the limit as $r\to \infty$ of $\sL(\cres(f;r))$ is
$\sA(f)$;  see Section~\ref{SubSec:LopsidedAmoeba} for more
information.

The main obstacle in turning Purbhoo's result into an efficient
approximation method for amoebas is the difficulty in computing the
polynomials $\cres(f;r)$. The degree and the number of
terms of $\cres(f;r)$ grow exponentially with $r$, but more
importantly, the methods used by computer algebra systems to find resultants fail to take
advantage of the sparseness of the polynomials $u_i^r-1$, and are therefore manifestly
inefficient when applied to $\cres(f;r)$.

Our main results are as follows.

\begin{enumerate}
\item We give a fast method to compute the cyclic resultant
   $\cres(f;2r)$ from $\cres(f;r)$ omitting all intermediate steps
   $\cres(f;r+1),\ldots,\cres(f;2r-1)$; see Section
   \ref{Sec:Resultants} for details.
We provide an experimental comparison of the runtimes using these quick resultants versus
   a general purpose resultant algorithm in Table \ref{Tab:Comparison}.
   We also give the complexity of our specialized algorithm for computing
  (certain) cyclic resultants, verifying our experimental conclusion that this method is significantly faster than the best resultant algorithm for real polynomials (Remark~\ref{rmk:ComplexityComparison}). 
 \item We provide an algorithm to approximate the unlog amoeba by semialgebraic sets, see Theorem~\ref{Thm:Semialgebraic} and Algorithm~\ref{Alg:Semialgebraic}.
\end{enumerate}

Purbhoo's Amoeba Approximation Theorem is the main ingredient in 
the proof of Theorem~\ref{Thm:Semialgebraic} and in Algorithm~\ref{Alg:Semialgebraic}. 
Using toric geometry, we also show that there exists a natural
correspondence between the boundary components of lopsided amoebas and
the boundary components of linear amoebas, where the latter are well understood
due to Forsberg, Passare, and Tsikh \cite{Forsberg:Passare:Tsikh}; see
Section \ref{Sec:Toric} for further details. 

Lemma \ref{lemma:differenceOfSquares} allows us to compute cyclic resultants for powers $r = 2^k$ using a divide and conquer algorithm. Note that several prominent algorithms are built on a similar approach, for example the Cooley--Tukey algorithm for the Fast Fourier Transformation \cite{Cooley:Tuckey}.

As a companion to this work we provide the first
implementation of lopsided approximation of amoebas,
available here:
\begin{quotation}
 {\footnotesize \url{http://www.math.tamu.edu/research/dewolff/LopsidedAmoebaApproximation/}.}
\end{quotation}
 We also provide the data presented in this article on this website. 
 The algorithms in this article are implemented in the
 computer algebra system \struc{\textsc{Singular}} \cite{Singular} and
 scripts to provide graphical outputs use the computer algebra
 system \struc{\textsc{Sage}} \cite{Sage}.

\subsection*{Outline}

This paper is organized as follows.
Section~\ref{Sec:Preliminaries} contains background on amoebas and
lopsided amoebas. Section~\ref{Sec:Resultants} outlines a fast way of
computing certain cyclic resultants. 
In Section~\ref{Sec:Lattice} we describe how the algorithm from
Section~\ref{Sec:Resultants} can be used to approximate amoebas.
Section~\ref{Sec:Toric} provides
a geometric interpretation of the lopsided amoeba as the intersection
of the amoeba of a toric variety and the amoeba of a hyperplane. 
In Section~\ref{Sec:ComputationSemialgebraic} we describe how to 
use results from Sections~\ref{Sec:Resultants} and \ref{Sec:Toric} to compute
semialgebraic sets approximating the amoeba.
Finally,
Section~\ref{Sec:Computation} is devoted to 
examples.

\subsection*{Acknowledgements}

We are very grateful to Alicia Dickenstein for her helpful
suggestions, especially on resultants.  We also thank Luis Felipe Tabera for his helpful comments. 
We are grateful to the anonymous referee for their helpful suggestions.

\section{Preliminaries}
\label{Sec:Preliminaries}

In this section, we review the theory of amoebas and
lopsided amoebas which is used in this work.

Recall that throughout this article, $f\in \CC[\bz^{\pm 1}]$ denotes the Laurent
polynomial \eqref{eqn:f}. We denote by $A$ the \emph{\struc{support}} of $f$,
namely, 
\[
\struc{A} \ := \ \big\{\alpha \in \ZZ^n \, \big|\, {\bf{z}}^\alpha \text{ appears with
  nonzero coefficient in } f\big\}.
\]
The \emph{\struc{Newton polytope}} of $f$, denoted ${\struc{\New(f)}}$,
is the convex hull in $\RR^n$ of the support set $A$ of $f$.

\subsection{Amoebas}
\label{SubSec:Amoeba}

The \textit{\struc{complement}} of the 
amoeba $\sA(f)$ is the set $\RR^n \minus \sA(f)$. The connected components of $\RR^n \minus \sA(f)$ are referred
to as the \emph{\struc{components of the complement of $\sA(f)$}}.

\begin{theorem}[{\cite[Chapter 6.1.A, Proposition~1.5, Corollary~1.8]{Gelfand:Kapranov:Zelevinsky}}]
For a nonzero Laurent polynomial $f$,  the complement\/ $\RR^n \minus \sA(f)$ is non-empty. Every
component of $\RR^n \minus \sA(f)$ is convex and open
(with respect to the standard topology).
\label{Thm:AmoebasAreClosedSets}
\end{theorem}

Forsberg, Passare and Tsikh~\cite{Forsberg:Passare:Tsikh} showed
that every component of the complement of an amoeba $\sA(f)$ 
can be associated with a specific lattice point in the Newton polytope $\New(f)$ via the \emph{\struc{order map}}:
\begin{eqnarray}
  \label{Equ:Ordermap}
\struc{\ord}\colon \R^n \minus \sA(f) & \to & \RR^n, \quad \mathbf{w} \mapsto (u_1,\ldots,u_n), \ \text{ where } \\
  u_j & := & \frac{1}{(2\pi \sqrt{-1})^n} \int_{\Log|\mathbf{z}| = \mathbf{w}} \frac{z_j \partial_j f(\mathbf{z})}{f(\mathbf{z})}
    \frac{dz_1 \cdots dz_n}{z_1 \cdots z_n} \ \text{ for all } \ 1 \le j \le n \, .\nonumber
\end{eqnarray}

\begin{theorem}[{\cite[Propositions~2.4~and~2.5]{Forsberg:Passare:Tsikh}}]
The image of the order map is contained in $\New(f) \cap \Z^n$. Let
$\mathbf{w},\mathbf{w}' \in \RR^n \minus \sA(f)$. Then $\mathbf{w}$ and
$\mathbf{w}'$ belong to the same component of the complement of
$\sA(f)$ if and only if $\ord(\mathbf{w}) = \ord(\mathbf{w}')$.
\label{Thm:OrderMap}
\end{theorem}

As a consequence of Theorem \ref{Thm:OrderMap}, it is possible to
define \struc{\textit{the component of order $\alp$}} of the complement of $\sA(f)$. 
We use the following notation:
\begin{eqnarray}
	\struc{E_{\alpha}(f)} & := & \{\mathbf{w} \in \R^n \minus
                                     \sA(f) \ \mid \ \ord(\mathbf{w})
                                     = \alpha\}. \label{Equ:DefEalp} 
\end{eqnarray} 

\subsection{Lopsided Amoebas}
\label{SubSec:LopsidedAmoeba}
We now give a more precise statement of the main
result in~\cite{Purbhoo}, which was alluded to in the introduction.

\begin{theorem}[{\cite[Theorem~1 ``rough version'']{Purbhoo}}]~
For $r \to \infty$ the family $\sL(\cres(f;r))$ converges uniformly to
$\sA(f)$. More precisely, for every $\eps > 0$ there exists an integer
$N$ such that for all $r>N$ the lopsided amoeba $\sL(\cres(f;r))$ is contained in an
$\eps$-neighborhood
of $\cA(f)$. 
\label{Thm:PurbhooMain}
\end{theorem}

\begin{remark}
\label{rem:Epsilon}
The integer $N$ in Theorem~\ref{Thm:PurbhooMain} 
depends only on $\eps$ and the Newton polytope (or degree) of $f$, 
and can be computed explicitly from this data. 
For simplicity, assume that $f$ is a polynomial of degree $d$.
Let $\Log|\mathbf{v}| \in \RR^n\setminus \sA(f)$ be of distance at least $\varepsilon$ from the
amoeba $\sA(f)$. Then, by \cite[Theorem 1]{Purbhoo}, if $r$ is chosen such 
that
\[
r \varepsilon \ \geq \ (n-1)\log r + \log\left((n+3)2^{n+1} d\right),
\]
then it holds that $\cres(f;r)$ is lopsided at the point $\Log|\mathbf{v}|$. We will typically be interested
in the case $r = 2^k$. Then, it suffices that
\[
\frac{2^k}{k} \ \geq \ \frac{C(n,d)}{\varepsilon},
\]
where $\struc{C(n,d)} := (n-1) \log 2 + \log\left((n+3)2^{n+1} d\right)$ is constant.
\end{remark}

If $f$ is lopsided at a point $\Log|\mathbf{v}| \in \R^n$, then it is straightforward to determine the order of the component of the complement of
$\cA(f)$ which contains $\Log|\mathbf{v}|$. 

\begin{theorem}[{\cite[Proposition~2.7]{Forsberg:Passare:Tsikh} and \cite[Proposition~4.1]{Purbhoo}}]
Let $A \subset \Z^n$ and let $f(\mathbf{z})$ be as in \eqref{eqn:f}.
Assume that 
$f(\bz)$ is lopsided at $\Log|\mathbf{v}|$ with dominating term  $|b_{\alp'}\,\mathbf{z}^{\alp'}|$.
Then $\ord(\Log|\mathbf{v}|) = \alp'$. 
In particular, if 
$\cres(f,r)$ is lopsided at $\Log|\mathbf{v}|$ with
dominating term $|b_{\alp'}\,\mathbf{z}^{\alp'}|$, then
$\ord(\Log|\mathbf{v}|) \cdot r^{n} = \alp'$. 
\label{Thm:DetermineOrder}
\end{theorem}

The following useful criterion for lopsidedness is a
consequence of the triangle inequality.

\begin{lemma}
\label{lemma:LopsidedCriterion}
The Laurent polynomial
$f(\zz)$ is not lopsided at $\Log|\zz|$ if and only if there exist
arguments $\theta \in (S^1)^d$ such that
$
\sum_{j = 1}^d b_j \,e^{\theta_j\sqrt{-1}}\,\mathbf{z}^{\alp(j)}  =  0.
$
\end{lemma}

\section{A Fast Algorithm to Compute Cyclic Resultants}
\label{Sec:Resultants}

In this section, we provide a fast way of computing certain
cyclic resultants.

We recall the \textit{\struc{Poisson formula}} for the resultant. If
$f(z)$ and $g(z)$ are 
univariate polynomials, and $g$ has leading coefficient $b$, then
\begin{eqnarray}
\label{eqn:Poisson}
\Res_z(f,g) & = & b^{\deg(f)} \prod_{\{\xi\, \mid \, g(\xi)=0\} } f(\xi).
\end{eqnarray}

\begin{lemma}
\label{lemma:Sparsity}
Let $g \in \CC[z]$ be a univariate polynomial of degree $\delta$ with leading
coefficient $b$, and let $r$ be a positive integer. Then
there exists a univariate polynomial $h$ of degree $\delta$ such that
\begin{eqnarray*}
\Res_u\big(g(zu),u^r-1\big) & = & h(z^r).
\end{eqnarray*}
\end{lemma}

\begin{proof}
By~\eqref{eqn:Poisson},
\begin{eqnarray*}
\Res_u\big(g(zu),u^r-1\big) & = & \prod_{j=0}^{r-1} g\big(e^{2\pi\sqrt{-1}j/r}z\big).
\end{eqnarray*}
Writing $g(z) = b \prod_{\ell=1}^\delta (z-\xi_\ell)$, we have
\begin{eqnarray*}
\Res_u\big(g(zu),u^r-1\big) & = & \prod_{j=0}^{r-1} \, b \prod_{\ell=1}^\delta
(e^{2\pi\sqrt{-1}j/r}z -\xi_\ell) \\
& = & b^r \prod_{\ell=1}^\delta \,  \prod_{j=0}^{r-1} (e^{2\pi\sqrt{-1}j/r}z
-\xi_\ell)  \\
& = & b^r \prod_{\ell=1}^\delta (z^r - \xi_{\ell}^r).
\end{eqnarray*}
Now use $h(z) = b^r \prod_{\ell=1}^\delta (z-\xi_\ell^r)$.
\end{proof}

\begin{lemma}
\label{lemma:RootShift}
Let $ h\big(z^{2^k}\big)= \Res_u\big(g(zu), u^{2^k}-1\big)$
be as in Lemma~\ref{lemma:Sparsity}. Then
\begin{eqnarray*}
\Res_u\big(g(zu),u^{2^k}+1\big) & = & h(-z^{2^k}).
\end{eqnarray*}
\end{lemma}

\begin{proof}
Note that the roots of $u^{2^k}+1$ are exactly the roots of $u^{2^k}-1$
multiplied by $e^{2\pi\sqrt{-1}/(2^{k+1})}$. Therefore,
by~\eqref{eqn:Poisson},
\begin{align*}
\Res_u\big(g(zu),u^{2^k}+1\big) & = \prod_{j=0}^{2^k-1}g\left(e^{2\pi\sqrt{-1}j/2^k}
e^{2\pi\sqrt{-1}/2^{k+1}}z\right) \\
& = h\left( (e^{2\pi \sqrt{-1} /2^{k+1}} z)^{2^k}\right)  \\
& = h\big( - z^{2^k}\big).\qedhere
\end{align*}
\end{proof}

\begin{lemma}
\label{lemma:differenceOfSquares}
The following identity holds.
\begin{eqnarray*}
\Res_u\big(g(zu),u^{2^{k+1}}-1\big) & = & \Res_u\big(g(zu),u^{2^k}-1\big) \cdot \Res\big(g(zu),u^{2^k}+1\big) .
\end{eqnarray*}
\end{lemma}

\begin{proof}
The equality follows from the Poisson formula, since $u^{2^{k+1}}-1 =
(u^{2^k}-1) (u^{2^k}+1)$, and these factors have disjoint sets of roots.
\end{proof}

\begin{algorithm}~
\label{CyclicResultantAlgorithm}
\begin{algorithmic}[1]
\State \INPUT{a polynomial $f \in \Q[\sqrt{-1}][\mathbf{x}]$ in $n$ variables of degree $d$ and $k \in \N$}
\State \OUTPUT{$\cres(f;2^k)$}
\State \COMMENT{Initialize a polynomial}
\State $\CycResult := f$
\State \COMMENT{Different variables are handled subsequently}
\For {$j = 1$ \textbf{to} $n$}
  \For {$l = 1$ \textbf{to} $k$}
    \State \COMMENT{Initialize multiplier polynomial}
    \State{$\Multiplier := \CycResult$}
    \State Change all signs of terms in multiplier, which have an exponent whose $j$-th entry is not divisible by  $2^l$.
    \State{$\CycResult = \CycResult \cdot \Multiplier$}
  \EndFor
\EndFor
\State \Return $\CycResult$
\end{algorithmic}
\label{Alg:QuickResultant}
\end{algorithm}
\begin{proof}[Proof of correctness of Algorithm~\ref{CyclicResultantAlgorithm}]
This follows from the definition of cyclic resultants \eqref{EquIterPolynLops},
Lemma~\ref{lemma:RootShift},
and Lemma~\ref{lemma:differenceOfSquares}.
\end{proof}

\begin{proof}[Complexity analysis of Algorithm~\ref{CyclicResultantAlgorithm}]
In order to perform the complexity analysis we will impose the 
assumption that $n=1$. This is not a severe restriction, as for a multivariate
polynomial our algorithm iterates the univariate case, however
it greatly simplifies the computations.

We count complexity as the number of arithmetic operations performed
in a field containing the coefficients of $f$. Here, we consider addition and 
multiplication to have the same cost. In the comparison with the signed
resultant algorithm we need to assume that the coefficients are real;
however, this does not affect our complexity analysis.

By Lemma~\ref{lemma:Sparsity} we have that, for each $l$,
the polynomial CycResult has at most $d+1$ terms. 
Multiplying the coefficient of every second term
(i.e., the terms whose exponents are not divisible by $2^l$)
by $-1$ requires $\lfloor d/2 \rfloor$ arithmetic operations.

Multiplying two univariate (real) polynomials of degree $d$ has complexity
$2 d^2 + 2 d + 1$, see \cite[Algorithm 8.3]{BPR03}. 

Since we perform this task $k$ times, the total complexity is less than
\[
k\left(\frac{d}2 +2 d^2 + 2 d + 1\right) = k\left(2 d^2 + \frac{5 d}2 + 1\right).
\]
Thus, the complexity is $O(kd^2)$.
\end{proof}

\begin{remark}
\label{rmk:ComplexityComparison}
Let $f$ be a real, univariate polynomial. The fastest general resultant
algorithm which the authors are aware of is the \emph{signed subresultant algorithm},
see, e.g., \cite[Algorithm 8.73]{BPR03}). It computes the resultant 
of two univariate polynomials of degrees $d_1$ and $d_2$ with arithmetic complexity
$O(d_1d_2)$. In our situation $d_1 = d$ and $d_2 = 2^k$; hence the signed
subresultant algorithm computes the cyclic resultant of order $k$ with
arithmetic complexity $O(d\,2^k)$. 

The algorithm we propose
is, hence, a vast improvement. In particular since,
in the typical situation, the degree $d$ is fixed while one varies the
parameter $k$ in order to obtain an improved approximation.
Overall, it reduces the runtime from exponential to polynomial in one variable and from double to single exponential in arbitrary many variables.
\end{remark}

We point out that Hillar and Levine~\cite{Hillar,Hillar:Levine} have
shown that the sequence of cyclic resultants
$\{\Res_z(g(z),z^k-1)\}_{k=1}^\infty$ satisfies a linear recurrence of order
$2^{\deg(g)}$. While this recurrence can be used to compute
\textit{all} cyclic resultants, it is in general slower than our method for
computing $\Res_u(g(zu),u^{2^{k}}-1)$,
 and the latter is sufficient for our purposes.

\section{Approximating an Amoeba Using Cyclic Resultants}
\label{Sec:Lattice}

In this section, we describe how Algorithm \ref{Alg:QuickResultant} can be applied to approximate an amoeba of a given polynomial. There are two obstacles in obtaining
an algorithm to approximate amoebas from the results of Section \ref{Sec:Resultants}. 
First, we can only compute cyclic resultants of finite level $k$.
That is, according to Theorem \ref{Thm:PurbhooMain} and Remark \ref{rem:Epsilon}, we can only test
for membership in some $\varepsilon$ neighborhood of the amoeba, where $\eps$
determines $k$. Second, it is not effective to test all points of $\RR^n$ 
for membership.

Neither of these obstacles are insurmountable. 
Our approach will be as follows. Let us fix some $\eps > 0$, which determines the integer $k$
in accordance with Remark~\ref{rem:Epsilon}.
Also, we will construct a grid $\struc{\cG} := [s,t]^n \cap \left(\ell \cdot \Z\right)^n \subseteq \Q^n$ for some $\ell \in \Q^n$.
Here, $\ell$ should be chosen small enough so that any ball of radius $\eps$ inside $[s, t]^n$ contains at least one point of $\cG$. Then, testing the finitely many points of $\cG$ on the level $k$, we are assured to
find points in every component of the complement of the amoeba whose intersection with $[s,t]^n$
contains a ball of radius $\eps$. 
Since the complement of the amoeba consists of a finite number of open sets, this ensures that we
will find all components of the complement if $\eps$ is chosen sufficiently small.
Though, we remark that there is currently no explicit expression describing how small
$\eps$ should be chosen.
In practice we will chose the grid $\cG$ rather than $\eps$; the latter can then be determined
from the former.

\begin{algorithm}~
\begin{algorithmic}[1]
\State \INPUT{a polynomial $f \in \Q[\sqrt{-1}][\mathbf{x}]$, start-/endpoints of grid $s,t$, steplength $\ell$ for grid, $k \in \N$}
\State \OUTPUT{List of gridpoints together with a bit indicating membership in $\eps$ neighborhood of
$\cA(f)$ and the order of the component of the complement if applicable}
\State \COMMENT{Create grid}
\State Create grid $\cG := [s,t]^n \cap (\ell \cdot \Z)^n$.
\State \COMMENT{Compute $2^j$-th cyclic resultants for $j = 1,\ldots,k$ via Algorithm \ref{Alg:QuickResultant}}
\For {$j = 1$ \textbf{to} $k$}
  \State{\State{$g_j := \CycResult(f,2^j)$}}
\EndFor
\State \COMMENT{Create empty list}
\State{$L := [\ ]$}
\State \COMMENT{Test gridpoints for lopsidedness}
\For {every $\mathbf{p}$ \textbf{in} $\cG$}
\State{IsLopsided $:=$ \texttt{false}}
\State{$j := 1$}
  \While{IsLopsided $=$ \texttt{false} \textbf{and} $j \leq k$}
    \If {$g_j$ is lopsided at $\mathbf{p}$}
      \State{IsLopsided $:=$ \texttt{true}}
    \Else{}
      \State{$j := j+1$}
    \EndIf
  \EndWhile
  \If {IsLopsided $=$ \texttt{true}}
    \State \COMMENT{Compute the order of $f$ at $\mathbf{p}$ via Theorem \ref{Thm:DetermineOrder}}
    \State{$\alp := \text{exponent of dominating term of } g_j \text{ at } \mathbf{p}$ divided by $2^j$}
    \State{Adjoin $(\mathbf{p},0,\alp)$ to $L$}
  \Else{}
    \State{Adjoin $(\mathbf{p},1,\text{NaN})$ to $L$}
  \EndIf
\EndFor
\State \Return $L$
\end{algorithmic}
\label{Alg:ApproximateAmoeba}
\end{algorithm}

\begin{proof}[Proof of Algorithm \ref{Alg:ApproximateAmoeba}]
Follows from Algorithm \ref{Alg:QuickResultant} and Theorem \ref{Thm:DetermineOrder}.
\end{proof}

We remark that computing the orders allows to determine from the finite list $L$
a set of polyhedra contained in the complement of the amoeba.
By Theorem \ref{Thm:AmoebasAreClosedSets}, the components of the complement of
an amoeba are convex.
Similarly, the recession cone of a component of the complement
with order $\alpha$ is equal to the normal cone of $\alpha$ in the Newton polytope
of $f$; see \cite{Passare:Tsikh:Survey}.

\section{Geometry of the Lopsided Amoeba}
\label{Sec:Toric}

The main result of this section,
Theorem~\ref{thm:LopsidedAmoebaIsIntersection}, gives a geometric interpretation of the lopsided
amoeba of a Laurent polynomial.

Let $f(\zz)$ be as in \eqref{eqn:f}, and denote by $\struc{\bb}$ the
coefficient vector $(b_1,\ldots,b_d)$.
We introduce the auxiliary $2d$-variate polynomial
\begin{eqnarray*}
\struc{F(\ww, \xx)} & := & \sum_{j = 1}^d w_j \,x_j ,
\end{eqnarray*}
which coincides with the polynomial $f(\zz)$ when evaluated at
$\ww=\bb$ and $x_j =
\zz^{\alpha(j)}$. 
The polynomial $F$ has amoeba $\sA(F)\subset \RR^{d}\times \RR^{d}$.
We denote by $\struc{\pr_2}$ the projection of $\RR^d\times \RR^d$
onto the second factor, and set ${\struc{F_\bb(\xx)}}:=F(\bb,\xx)$.

For convenience, we denote the coordinates of $\RR^{d}\times \RR^{d}$
by $(\Log |\mathbf{w}|,\Log |\mathbf{x}|)$. 

\begin{lemma}
Let $\struc{H_\bb}$ be the $d$-dimensional
affine space defined by $\Log |\ww| = \Log |\bb|$.
Then  $\sA(F_\bb) = \pr_2\left(\sA(F) \cap H_\bb \right)$.
\end{lemma}

\begin{proof}
By Lemma~\ref{lemma:LopsidedCriterion},
$(\Log |\bb|, \Log |\xx|)\in \sA(F)$ if and only if there
exist arguments
$(\theta_\bb, \theta_\xx)\in (S^1)^d \times (S^1)^d$ such that
$F(\bb e^{\sqrt{-1}\theta_\bb}, \xx e^{\sqrt{-1}\theta_\xx}) = 0$.
We have
\[
F(\bb \,e^{\sqrt{-1}\theta_\bb}, \xx \,e^{\sqrt{-1} \theta_\xx}) \ = \ F(\bb, \xx \,e^{\sqrt{-1} (\theta_\xx+ \theta_\bb)}) \ = \ F_\bb(\xx \,e^{\sqrt{-1} (\theta_\xx+ \theta_\bb)}),
\]
and the result follows.
\end{proof}

Recall that the
\struc{\textit{affine toric variety} $X_A$} associated to 
$A=\{\alpha(1),\dots,\alpha(d)\}\subset \ZZ^n$,
is the Zariski closure in $\CC^d$ of the image of the monomial map
$(\CC^*)^n \to (\CC^*)^d$ given by 
$\mathbf{z} \mapsto (\mathbf{z}^{\alpha(1)},\dots,\mathbf{z}^{\alpha(d)})$.

The amoeba $\struc{\sA(X_A)}$ of the toric variety $X_A$
is a linear subvariety of $\RR^d$, parameterized by 
$\Log |\zz| \mapsto (\alpha(1) \cdot \Log |\zz|,\dots,\alpha(d) \cdot
\Log|\zz|)$. 
We denote by $\struc{s_A}\colon \sA(X_A) \rightarrow \RR^n$
the inverse of this mapping.

We embed $\sA(X_A)$ in the second factor $\RR^d\times \RR^d$, in order
to be able to compare this set to
the amoeba $\sA(F)$. 
Let $\struc{\overline{\pr}_A} = s_A\circ \pr_2$. Note that, given
$\bb$, the induced map 
\[
\struc{\overline{\pr}_A}\colon H_\bb \cap (\RR^d \times \sA(X_A))
\rightarrow \RR^n 
\]
is an affine isomorphism.

\begin{theorem}
\label{thm:LopsidedAmoebaIsIntersection}
We have that $\sL(f) = \overline{\pr}_A\left(\sA(F)\cap H_\bb \cap 
(\RR^d \times \sA(X_A))
\right)$.
\end{theorem}

More precisely, if we identify $\RR^n$ with $H_\bb \cap (\RR^d \times
\sA(X_A)) \subset \RR^d\cap \RR^d$ 
by the map $\overline{\pr}_A$, then $\sL(f)$ is identified with $\sA(F)$
(or rather, $\sL(f)$ is identified with the part of $\sA(F)$ which is
contained in $H_\bb \cap  (\RR^d \times \sA(X_A))$).

\begin{proof}
Let $\bb$ be fixed. 
Since $\overline{\pr}_A\colon H_\bb \cap (\RR^d \times \sA(X_A))
\rightarrow \RR^n$ is an affine isomorphism, it has an inverse 
$\overline{\pr}_A^{-1}$ given by
\[
\overline{\pr}_A^{-1}\colon \RR^n \rightarrow H_\bb \cap (\RR^d \times
\sA(X_A)), 
\qquad
\Log|\zz| \mapsto \left(\Log|\bb|, \, \alp(1) \cdot \Log|\zz|,\dots,
  \alp(d) \cdot \Log|\zz| \right).
\]
By Lemma~\ref{lemma:LopsidedCriterion}, $f$ is not lopsided at
$\Log|\zz|$ if and only if there exist 
angles $\theta_\bb \in (S^1)^d$ such that 
\begin{eqnarray*}
\sum_{j = 1}^d b_j e^{\theta_{b_{j}}\,\sqrt{-1}}\,\mathbf{z}^{\alp(j)} & = & 0,
\end{eqnarray*}
which is equivalent to $F(\bb e^{\theta_\bb \sqrt{-1}}, \zz) = 0$.
But, the latter is equivalent to
$
\overline{\pr}_A^{-1}(\Log(\zz)) \ \in \ \sA(F),
$
and the proof is complete.
\end{proof}

\section{Approximating Unlog Amoebas by Semi-algebraic Sets}
\label{Sec:ComputationSemialgebraic}

Recall that the unlog amoeba $\sU(f)$ is a semi-algebraic set. In this section we give an algorithm which provides an approximation
of the unlog amoeba $\sU(f)$ as a sequence of semi-algebraic sets.

The \textit{\struc{boundary}} of the amoeba $\sA(f)$ is denoted by
$\struc{\partial \sA(f)}$; similarly, the boundary of a
lopsided amoeba $\sL(f)$ by is denoted by $\struc{\partial
  \sL(f)}$. Our computation is based on the following statement, which
uses the notation introduced in~\eqref{Equ:DefEalp}.

\begin{theorem}
Let $f(\mathbf{z})$ be as in \eqref{eqn:f} and let $\cres(f;r) =  \sum_{\beta \in
  A(r)} b_{r,\beta} \,\mathbf{z}^{\beta}$. Then, the boundary $\partial \sL(\cres(f;r))$
converges to the boundary $\partial \sA(f)$ as $r \to \infty$, in the same sense as in
Theorem~\ref{Thm:PurbhooMain}. In particular, if\/
$\RR^n \minus \sA(f) = \bigcup_{j = 1}^k E_{\alp(j)}(f)$ and if $r$ is
sufficiently large, then the semi-algebraic set $\sU(f)$ belongs to an $\varepsilon$-neighborhood of
the semi-algebraic set
\begin{eqnarray}
\begin{array}{rcl}
  \sum_{\beta \in A(r) \setminus \{r^{n} \cdot \alp(1)\}}
  |b_{r,\beta}| \mathbf{x}^{\beta} & \geq &
|b_{r,r^{n} \cdot \alp(1)}|
\cdot \mathbf{x}^{r^{n} \cdot \alp(1)}, \\
 \vdots & \vdots & \vdots \\
 \sum_{\beta \in A(r) \setminus \{r^{n} \cdot \alp(k)\}} |b_{r,\beta}|
 \mathbf{x}^{\beta} & \geq &
|b_{r,r^{n} \cdot \alp(k)}|
\cdot \mathbf{x}^{r^{n} \cdot \alp(k)}, \\
 x_1,\ldots,x_n & \geq & 0.
\end{array}
\label{Equ:SemialgebraicDescription}
\end{eqnarray}
\label{Thm:Semialgebraic}
\end{theorem}

\begin{proof}
Since each connected component of the complement of the set defined by 
\eqref{Equ:SemialgebraicDescription} is contained in the complement of $\sU(f)$,
we have that the Hausdorff distance of $\partial \sL(\cres(f;r))$ and $\partial \sA(f)$
is dominated by the Hausdorff distance of $\sL(\cres(f;r))$ and $\sA(f)$.
In particular, the statement regarding convergence of the boundaries follows
immediately from Theorem~\ref{Thm:PurbhooMain}. Using
Theorem~\ref{thm:LopsidedAmoebaIsIntersection}, we see that a
component of the complement of $\sL(\cres(f;r))$ of order $\alp'$ is
given by $\Log|\cdot|$ image of the semi-algebraic set  
\begin{eqnarray}
 |b_{r,\alp'}| \mathbf{x}^{\alp'} & > & \sum_{\beta \in A(r) \setminus \{\alp'\}} |b_{r,\beta}| \mathbf{x}^{\beta}, \label{Equ:ProofSemialgebraic} \\
 x_1,\ldots,x_n & \geq & 0.  \nonumber
\end{eqnarray}
By Theorem~\ref{Thm:PurbhooMain} we know that for all sufficiently
large $r$ there exists a bijection between the set of all components
of the complement of $\sA(f)$ and the set of all components of the
complement of $\sL(\cres(f,r))$. By Theorem~\ref{Thm:DetermineOrder}
it follows that if a component of the complement of $\sA(f)$ has order
$\alp'$, then the corresponding component of the complement of
$\sL(\cres(f;r))$ has order $\alp' \cdot r^n$. This
and~\eqref{Equ:ProofSemialgebraic} complete the proof. 
\end{proof}

This theorem allows us to compute a sequence of semi-algebraic sets approximating
the unlog amoeba $\sU(f)$ as follows: 

\begin{algorithm}~
\label{alg:SemiAlgebraicSet}
\begin{algorithmic}[1]
\State \INPUT{$f \in \Q[\sqrt{-1}][\mathbf{x}]$, $\varepsilon > 0$}
\State \OUTPUT{An approximation of $\sU(f)$ by a semi-algebraic set.}
\State Choose $r \in \NN$ sufficiently large with respect to $\varepsilon$ according to Remark~\ref{rem:Epsilon}.
\State \COMMENT{Compute all potential images of $\ord(\sA(f))$}
\State $M := \New(f) \cap \Z^n$
\State \COMMENT{Define a zero polynomial}
\State $g := \mathbf{0}$
\State \COMMENT{Construct $g$ as sum of abs. values of terms of $\cres(f;r)$}
\For {$b_\alp \mathbf{x}^{\alp}$ term \textbf{in} $\cres(f;r)$}
  \State $g := g + |b_\alp| \mathbf{x}^{\alp}$
\EndFor
\State \COMMENT{Construct inequalities for semi-alg. description of $\sU(f)$}
\For {$\alp / r^n \in M$}
\State $g(\alp) := g - 2 \cdot |b_\alp| \mathbf{x}^{\alp}$
\EndFor
\State \COMMENT{Construct semi-alg. description of $\sU(f)$}
\State $L := \bigcup_{j = 1}^n \{x_j \geq 0\} \cup \bigcup_{\alp / r^n \in M} \{g(\alp)\}$
\State \Return $L$
\end{algorithmic}
\label{Alg:Semialgebraic}
\end{algorithm}

\begin{proof}[Proof of correctness of Algorithm~\ref{alg:SemiAlgebraicSet}]
This follows from Theorem~\ref{Thm:PurbhooMain} and Remark~\ref{rem:Epsilon}.
\end{proof}

Note that it is sufficient to consider the image of the order map of $\sA(f)$ as the set $M$ in Algorithm
\ref{Alg:Semialgebraic}, if the image of the order map is known. For
example, if $\supp(f)$ is a \textit{\struc{circuit}}, that is, a
minimal affinely dependent set, then it is well-known that the image
of the order map is at most $\supp(f)$; see~\cite[Theorem 12, Page
36]{Rullgard:Diss}.

\section{Experimental Comparisons and Computations}
\label{Sec:Computation}

In this section we show experimentally that our fast method for
computing cyclic resultants allows us to efficiently tackle the membership
problem, as well as approximate an unlog
amoeba using semi-algebraic sets. We compare this approach to a general purpose 
resultant algorithm: the comparison of the runtimes of both methods in
given in Table~\ref{Tab:Comparison}.

\subsection{Outline, Input, and Assumptions}

As we mentioned in the Introduction, the two methods we use in this
article to compute the amoeba of a Laurent polynomial are:
\begin{enumerate}
 \item an approximation based on solving the Membership Problem~\ref{Prob:MembershipProblem}, and
 \item an approximation given by semi-algebraic sets converging to the unlog amoeba $\sU(f)$.
\end{enumerate}
We use the computer algebra system \textsc{Singular} \cite{Singular}
for the computation of resultants and for the lopsidedness test. 
We produce pictures by
transferring our \textsc{Singular} output to the computer algebra
system \textsc{Sage} \cite{Sage}.

For technical reasons we 
encode, in \textsc{Singular}, coefficients over a rational ring with $n$ variables which, in
order to mimic complex numbers, has a parameter $\struc{I}$ with a
minimal polynomial $I^2 + 1$. 

When computing $\cres(f;2^k)$, we refer to the number $k$ as the
\struc{\emph{level}} of the resultant.
In theory, our \textsc{Singular} code can be used to compute cyclic
resultants of an arbitrary level, and also to test lopsidedness of a
single point; see Sections
\ref{SubSec:ComputationComparisonOfResultants} and
\ref{SubSec:ComputationMembership}. 
In order to produce
pictures, we need to restrict our computations to a grid as described in Section~\ref{Sec:Lattice}.
Depending on the approach used, we handle this issue differently;
see Sections \ref{SubSec:ComputationMembership} and
\ref{SubSec:ComputationSemialgebraic}. 

\subsection{The Computation of Cyclic Resultants}
\label{SubSec:ComputationComparisonOfResultants}

In order to compute amoebas via membership tests or via a sequence of 
semi-algebraic sets converging to the unlog amoeba $\sU(f)$ for a given
Laurent polynomial $f$ we need to compute the $r$-th cyclic resultant $\cres(f;r)$. 
The larger we choose $r \in \N^*$ the more accurate we expect the approximation to be,
 but also more resources are required, and the output obtained increase in complexity.
 Indeed, $\cres(f;r)$ is typically extremely large by any measure. For instance, consider
$f(z_1,z_2) = z_1^3+z_1z_2+z_2^3+1$; a polynomial in two
variables of degree three with four terms. Table~\ref{Tab:SizeCyclicResultants}
provides a rough view of the size of $\cres(f;r)$ for different values of $r$. 
\renewcommand{\arraystretch}{1.5}

\begin{table}
$$
\begin{array}{|l||c|c|c|c|c|c|} \hline
  r & 2 & 4 & 8 & 16 & 32 & 64 \\ \hline \hline
 \text{number of terms} & 10 & 31 & 109 & 409 & 1585 & 6241 \\ \hline
 \text{degree} & 12 & 48 & 192 & 786 & 3072 & 12288\\ \hline
 \text{coefficient magnitude} & 9 & 860 & > 10^{12} & > 10^{51} & 10^{204} & 10^{811} \\ \hline
\end{array}
$$
\caption{The size of $\cres(z_1^3+z_1z_2+z_2^3+1;r)$ with $r = 2^k$.}
\label{Tab:SizeCyclicResultants}
\end{table}

In Table~\ref{Tab:Comparison} we provide a comparison 
of the runtimes of our resultant algorithm and
a standard computation using a built in resultant command
for the following polynomials with real and complex coefficients in two and three variables: 
\begin{eqnarray*}
f_1(z_1,z_2) & = & z_1^3+z_1z_2+z_2^3+1, \\
f_2(z_1,z_2) & = & (5+\sqrt{-1}) \cdot z_1^3+ \sqrt{-1} \cdot z_1z_2+(4+\sqrt{-1}) \cdot z_2^3+1, \\
f_3(z_1,z_2,z_3) & = & z_1^4z_2+z_1z_2z_3^5+z_1^2z_2^4+z_1z_2^2+z_1z_2z_3+z_1z_2z_3^3+1. 
\end{eqnarray*}

In the table, \texttt{Level} denotes the level of the cyclic resultant.
\texttt{Runs} denotes the number of test runs we
made for the particular polynomial and level. \texttt{Quick resultant}
denotes the average runtime (in seconds) for the computation of cyclic
resultants using our new recursive method,
while \texttt{Iterated resultant} denotes the average runtime (in
seconds) of a procedure using~\eqref{EquIterPolynLops} and the {\sc Singular} \texttt{resultant} command. Finally,
\texttt{Factor} denotes the ratio between the latter and former runtimes.

For instances, which can be computed very quickly, we took up to $100$ test runs and computed the average
runtime of both algorithms \footnote{All computations were performed
  with \textsc{Singular} \texttt{Version 3-0-4} and on a Desktop
  computer, Distribution: \texttt{Linux Mint 13 Maya}, Kernel:
  \texttt{3.8.0-44-generic x86\_64}, CPU: \texttt{Quad core Intel Core
    i7-2600 CPU}, RAM: \texttt{16GB}.}. 

\begin{table}
$$
\begin{array}{|c|c|c|D{.}{.}{4}|D{.}{.}{4}|D{.}{.}{2}|} \hline
 \text{Polynomial} & \text{Level} & \text{Runs} & \multicolumn{1}{c|}{\text{Quick resultant}} & \multicolumn{1}{c|}{\text{Iterated resultant}} & \multicolumn{1}{c|}{\text{Factor}}  \\ \hline \hline
 f_1 & 3 & 100 & 0.0062 & 0.0263 & 4.24 \\ \hline
 f_1 & 4 & 100 & 0.0520 & 1.7120 & 32.92 \\ \hline
 f_1 & 5 & 5 & 8.0400 & 221.6040 & 27.56 \\ \hline
 f_1 & 6 & 1 & 28.1400 & 118397.8600 & 4207.46 \\ \hline
 f_2 & 3 & 100 & 0.0354 & 0.6680 & 18.87 \\ \hline
 f_2 & 4 & 100 & 0.3952 & 88.9828 & 301.42 \\ \hline
 f_2 & 5 & 3 & 5.0433 & 67927.7433 & 13468.91 \\ \hline
 f_3 & 3 & 10 & 0.0080 & 245.5500 & 30693.75 \\ \hline
\end{array}
$$
\caption{A comparison of runtimes for the computation of cyclic resultants.}
\label{Tab:Comparison}
\end{table}

\subsection{Solving the Membership Problem via Lopsidedness}
\label{SubSec:ComputationMembership}

In this section we describe how amoebas can be approximated in
practice by solving the membership problem using lopsidedness
certificates. Again, we use the following polynomial as a running
example  
\begin{eqnarray}
 f(z_1,z_2) & = z_1^3+z_1z_2+z_2^3+1. \label{Equ:RunningExample}
\end{eqnarray}
The level two cyclic resultant equals
(where we point out that, for notational convenience, 
the variables are denoted \texttt{x,y,z} in {\sc Singular} output):
\begin{lstlisting}[basicstyle=\ttfamily\footnotesize]
x48+28*x40y4-4*x36y12-4*x36+246*x32y8-156*x28y16-156*x28y4+
6*x24y24+576*x24y12+6*x24-860*x20y20-860*x20y8-156*x16y28+
969*x16y16-156*x16y4-4*x12y36+576*x12y24+576*x12y12-4*x12+
246*x8y32-860*x8y20+246*x8y8+28*x4y40-156*x4y28-156*x4y16+
28*x4y4+y48-4*y36+6*y24-4*y12+1
\end{lstlisting}

In order to compute the intersection of $\sA(f)$ with $[s,t]^n
\subseteq \R^n$, the idea is to define a grid
on $[s,t]^n$, and to test $f$ and $\cres(f;r)$ for lopsidedness on
every grid point up to a pre-defined $r$. According to
Theorem~\ref{Thm:PurbhooMain} all grid points $\ww \in \R^n$
which do not belong to $\sA(f)$ are lopsided for $\cres(f;r)$ if $r$ is
sufficiently large. 
For our running example we choose the grid given by $[-2,2]^2 \cap (\frac{1}{10} \Z)^2 \subseteq \R^2$.  

\begin{figure}
\includegraphics[width=0.45\linewidth]{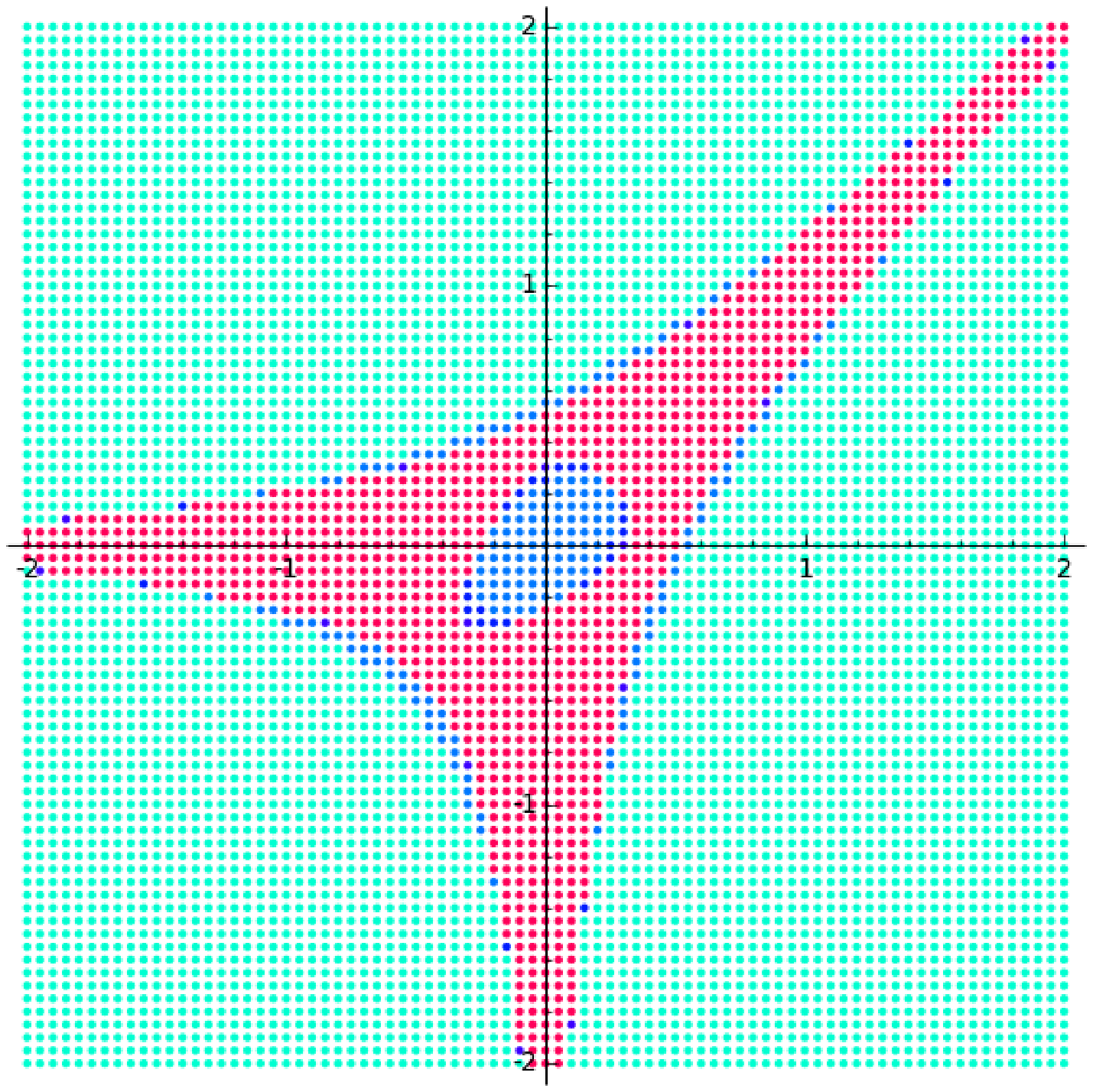}
\includegraphics[width=0.45\linewidth]{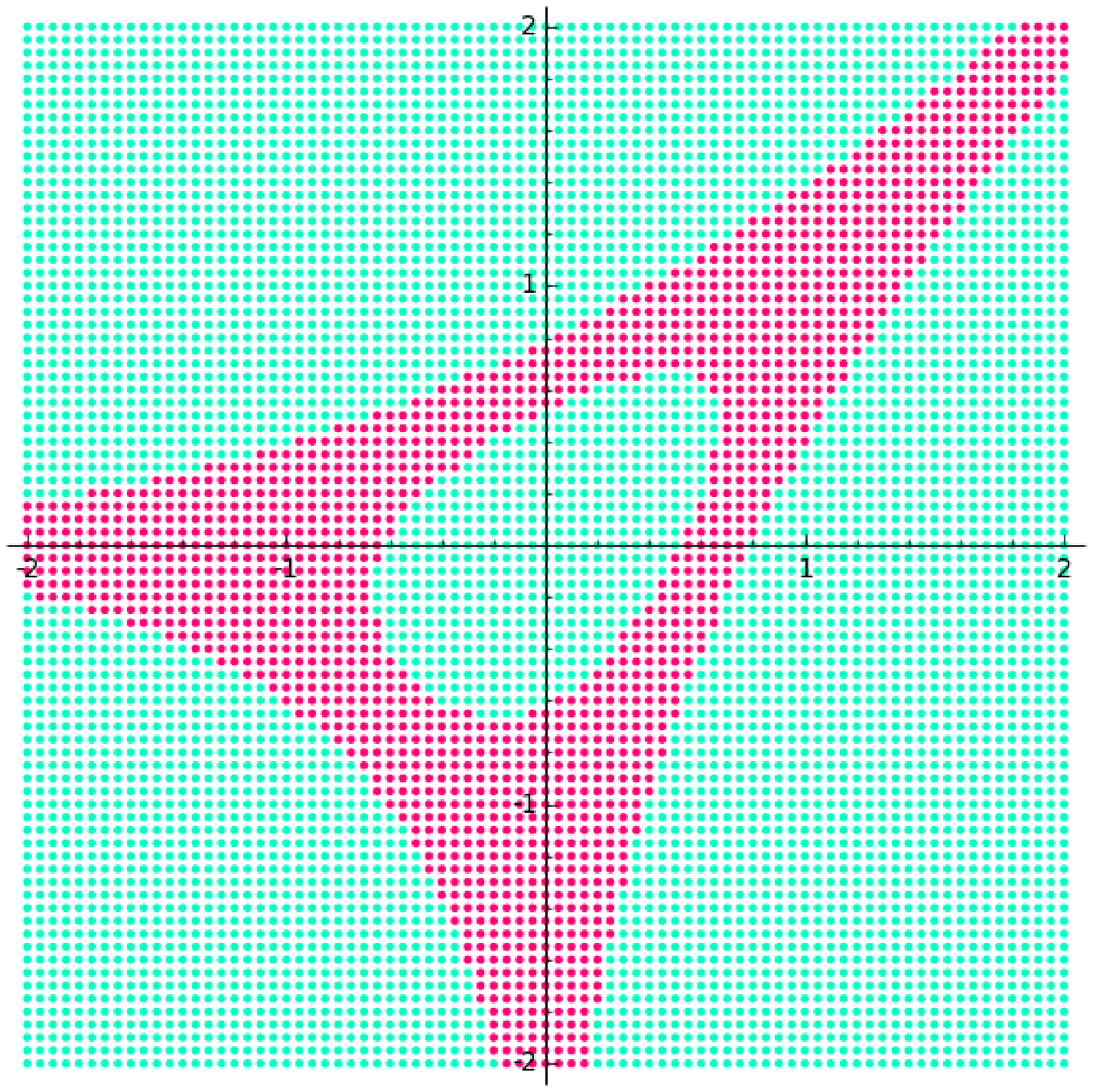}
\caption{Plot of the amoeba of the polynomial $f(z_1,z_2) = z_1^3 + z_2^3 +
  b\cdot z_1z_2 +1$ for $b = 2$ and $b = -4$ on a $[-2,2]^2$ grid with $81
  \times 81$ 
  grid points and relaxation level at most $4$. On the left hand side
  the turquoise points are the complement of the lopsided amoeba of $f$, light blue
  points are lopsided on relaxation level $3$ and dark blue points (very
  few; mainly at the boundary of hole in the middle) at relaxation level $4$. Red
  points were  never lopsided and hence are presumed to belong to the
  amoeba of $f$. On the right hand side the amoeba coincides with the lopsided amoeba in this plot.}
\label{Fig:AmoebaMembershipPlot}
\end{figure}

For every point $\ww$ on the grid, we need to evaluate every monomial
of $f$ or $\cres(f;r)$ at $\Log^{-1}(\ww) =
(\exp(w_1),\ldots,\exp(w_n))$. 
We test all points in our list for lopsidedness 
using $f = \cres(f;2^0)$ and $\cres(f;2^1),\ldots,\cres(f;2^k)$; here $k =
4$ .
The algorithm
successively searches for lopsidedness certificates for $\ww
\notin \sL(\cres(f;j))$ for $j \in \{0,\ldots,k\}$.
If a certificate
is found, then we save the level $j$. Otherwise $j$ is increased by
$1$. If $\cres(f;r)$ is not lopsided at $\ww$, then
we presume that $\ww \in \R^n$ is a point in $\sA(f)$ since no
certificate was found. 
This does not preclude further testing being performed, that may eliminate some of these points at a higher level.

\subsection{Approximating Unlog Amoebas using Semi-algebraic Sets}
\label{SubSec:ComputationSemialgebraic}

In this last section we demonstrate how to use our software to
obtain an approximation of the unlog amoeba $\sU(f)$ for a given $f \in \C[\mathbf{z}^{\pm 1}]$ using semi-algebraic sets. Moreover, we can draw an
implicit plot of (the boundary of) the image of this semi-algebraic description. 

Returning to our running example \eqref{Equ:RunningExample}
we compute the semi-algebraic description of $\sU(f)$ in
\textsc{Singular}. In this example, we use levels $1,2$, and $3$
given by the first three non-trivial instances of
\texttt{quickcyclicresultant} for $f$. 
The output is shown in
Figure~\ref{Fig:UnlogAmoebaSemialgebraicPlot}. An interesting
observation is that, according to this plot, the approximation of the
sought semi-algebraic description of $\sU(f)$ does not necessarily improve
monotonically. One can see that there exist certain areas where the
\texttt{level} $2$ approximation is better than the \texttt{level} $3$
approximation. This behavior does not seem to be a plotting artifact,
and raises theoretical questions on the nature of the convergence of
this approximation.

\begin{figure} 
\includegraphics[width=0.5\linewidth]{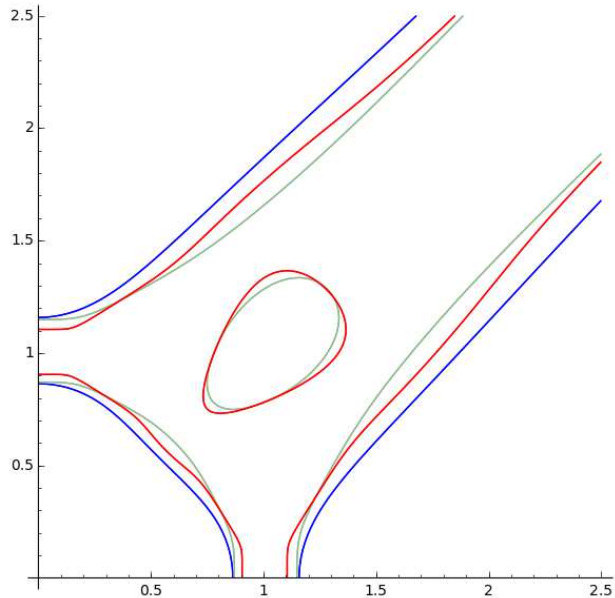}
\caption{A plot of the semi-algebraic approximation of the unlog amoeba
  $\sU(f)$ for $f(z_1,z_2) = z_1^3+z_2^3+2\cdot z_1z_2+1$ using relaxed
  polynomials given by \texttt{quickcyclicresultant} on level 1
  (blue), 2 (dark green), and 3 (red).} 
\label{Fig:UnlogAmoebaSemialgebraicPlot}
\end{figure}

\bibliographystyle{amsalpha}
\bibliography{lopsided}

\end{document}